\documentclass[11pt]{amsart}

\usepackage[utf8]{inputenc}
\usepackage[T1]{fontenc}
\usepackage{amsmath,amssymb,amsthm,mathtools}
\usepackage{mathrsfs,bm}
\usepackage{enumitem,booktabs,array,tabularx}
\usepackage{hyperref,cleveref}
\usepackage[a4paper,margin=1in]{geometry}
\usepackage{forest}
\usepackage[section]{placeins}
\hypersetup{
  colorlinks=true,
  linkcolor=blue,
  citecolor=blue,
  urlcolor=blue,
  pdftitle={Formal Grammars and Generating Functions},
  pdfauthor={Zhijun Wang}
}

\newtheorem{theorem}{Theorem}[section]

\newtheorem{corollary}[theorem]{Corollary}
\theoremstyle{definition}
\newtheorem{definition}[theorem]{Definition}

\newcommand{\Z}{\mathbb{Z}}
\newcommand{\Q}{\mathbb{Q}}

\newcommand{\C}{\mathbb{C}}
\newcommand{\RE}{\mathbf{RE}}
\newcommand{\Rec}{\mathbf{Rec}}

\title[Formal Grammars and Generating Functions]{Formal Grammars and Generating Functions:\\
From Algebraic Structure to Non-Algebraic and Non-Derivable Series}
\author{Zhijun Wang}
\address{St Peter's College, Department of Engineering Science\\
University of Oxford\\
Oxford, UK}
\email{zhijun.wang@ox.ac.uk}
\date{}

\begin{document}

\begin{abstract}
This paper studies the relationship between formal grammars, recursive
production systems, computability, and ordinary length generating functions.
We first recall the classical Chomsky--Schutzenberger theorem and the
commutative translation usually associated with the DSV method. We then
distinguish DSV systems satisfying a finite-algebra condition from more general
grammar-to-generating-function constructions. In particular, the language
$\{a^n b^n c^n:n\geq1\}$, although not context-free, is indexed and has the
rational length generating function
$z^3/(1-z^3)$. We show that the finite-algebra condition is sufficient for a polynomial
DSV system to have an algebraic target generating function. The converse is false: the
same rational generating function can arise from languages of radically
different grammatical complexity.

We next establish a computability principle. If a recursively enumerable
language has a computable length-counting sequence, then the language itself
is recursive. Consequently, for every non-recursive recursively enumerable
language, the coefficients of its ordinary length generating function form a
non-computable sequence. We further deduce that a recursively enumerable
language with a $D$-finite length generating function must be decidable.

For recursive languages, by contrast, the coefficient sequence is always
computable, independently of whether the generating function is algebraic,
D-finite, or transcendental.

Finally, we examine singularities of generating functions with non-negative
integer coefficients. We relate productive recursion to infinite languages
and finite radii of convergence, and review the Polya--Carlson dichotomy,
Hadamard gap phenomena, algebraic functions, D-finite transcendental
functions, and Siegel's $G$-functions.

The resulting picture is inherently asymmetric:
\[
\text{grammar}\longrightarrow\text{functional equations}
\longrightarrow\text{generating function}\longrightarrow\text{analytic structure},
\]
whereas the reverse implications generally require additional information
not contained in the ordinary length generating function.
\end{abstract}

\maketitle
{\small\tableofcontents}

\section{Introduction and Preliminaries}

Let $L\subseteq\Sigma^*$ be a formal language over a finite alphabet
$\Sigma$. Its ordinary length generating function is
\begin{equation}
F_L(z)=\sum_{n\geq0}\ell_n z^n,
\qquad
\ell_n=\#\{w\in L:|w|=n\}.
\label{eq:language-ogf}
\end{equation}

The map
\[
L\longmapsto F_L
\]
forgets the identity of individual words and retains only the number of
accepted words of each length. Consequently, very different languages may
have exactly the same ordinary generating function.

We use the following algebraic and differential classifications throughout.

\begin{definition}[Algebraic generating function]
A formal power series $F(z)\in\Q[[z]]$ is \emph{algebraic over $\Q(z)$}
if there exists a nonzero polynomial $P(z,y)\in\Q[z,y]$ such that
\[
P(z,F(z))=0.
\]
Otherwise, $F(z)$ is \emph{transcendental over $\Q(z)$}.
\end{definition}

\begin{definition}[$D$-finite generating function]
A formal power series $F(z)\in\Q[[z]]$ is \emph{$D$-finite over $\Q(z)$}
if there are polynomials $p_0(z),\ldots,p_r(z)\in\Q[z]$, with
$r\geq0$ and $p_r\ne0$, such that
\[
p_r(z)F^{(r)}(z)+p_{r-1}(z)F^{(r-1)}(z)
+\cdots+p_0(z)F(z)=0.
\]
All algebraic power series over $\Q(z)$ are $D$-finite, but the
converse need not hold.
\end{definition}

The guiding question is which implications survive the passage
\[
\text{grammatical specification}\longrightarrow
\text{length-counting sequence}\longrightarrow F_L(z)
\longrightarrow\text{analytic behavior}.
\]
The passage from a language to its length-counting sequence is many-to-one:
word identities, parsing information, and membership algorithms are not
recoverable in general. Moreover, these arrows describe mathematical
associations, not a universal effective algorithm for extracting coefficients
from arbitrary grammars. Sections~2--5 distinguish the hypotheses under
which stronger conclusions are available.

\section{From Grammars to Generating Functions}

\subsection{The classical theorem}

We start with the classical algebraicity theorem.

\begin{theorem}[Chomsky--Schutzenberger]
\label{thm:CS}
Let $L\subseteq\Sigma^*$ be generated by an unambiguous context-free
grammar. Define
\[
\ell_n=\#\{w\in L:|w|=n\},
\qquad
F_L(z)=\sum_{n\geq0}\ell_nz^n.
\]
Then $F_L(z)$ is algebraic over $\Q(z)$.
\end{theorem}

This theorem is one of the fundamental bridges between formal-language theory
and algebraic generating functions
\cite{ChomskySchutzenberger,Flajolet,FlajoletSedgewick}.

The proof proceeds by replacing the productions of an unambiguous
context-free grammar by algebraic equations for the generating functions
associated with the nonterminals. Concatenation becomes multiplication,
disjoint union becomes addition, and terminals contribute powers of the
length variable. Elimination of the auxiliary generating functions then
produces a polynomial equation
\[
P(z,F_L(z))=0.
\]
This commutative translation is usually referred to as the DSV method.

\subsection{A finite-algebra criterion for DSV systems}

A finite list of polynomial equations need not force every formal-power-series
solution to be algebraic: the equations could leave some variables
unconstrained. The following additional condition is sufficient.

\begin{definition}[Finite-algebra condition]
Let $K=\Q(z)$ and let
\[
P_1(z,Y_1,\ldots,Y_m)=\cdots=P_s(z,Y_1,\ldots,Y_m)=0
\]
be polynomial equations with $P_i\in\Q[z,Y_1,\ldots,Y_m]$.
Write
\[
I_K=\langle P_1,\ldots,P_s\rangle
\subseteq K[Y_1,\ldots,Y_m].
\]
We say that the system satisfies the \emph{finite-algebra condition} if
\[
\dim_K\!\bigl(K[Y_1,\ldots,Y_m]/I_K\bigr)<\infty.
\]
\end{definition}

\begin{theorem}[Finite-algebra criterion for a grammar-derived system]
\label{thm:grammar-finite-algebra}
Suppose the DSV translation of a grammatical specification yields a
polynomial system as above with a specified formal-power-series solution
$(Y_1(z),\ldots,Y_m(z))\in\Q[[z]]^m$ and target
$F_L(z)=Y_1(z)$. If the system satisfies the finite-algebra condition,
then $F_L(z)$ is algebraic over $\Q(z)$.
\end{theorem}

\begin{proof}
Let $\mathcal A=K[Y_1,\ldots,Y_m]/I_K$, and let $y_1$ be the residue
class of $Y_1$ in $\mathcal A$. Since $\mathcal A$ is finite-dimensional
over $K$, multiplication by $y_1$ defines a $K$-linear operator
$\mu_{y_1}:\mathcal A\to\mathcal A$. Its characteristic polynomial
$q(t)\in K[t]$ is monic and nonzero. By the Cayley--Hamilton theorem,
$q(\mu_{y_1})=0$; evaluating at the class of $1$ yields
$q(y_1)=0$ in $\mathcal A$.

The specified solution satisfies each polynomial equation $P_i=0$.
Substituting $Y_j=Y_j(z)$ therefore gives $q(Y_1(z))=0$ in the
field $\Q((z))$ of formal Laurent series. Clearing the denominators of
$q$ yields a nonzero polynomial $Q(z,t)\in\Q[z,t]$ such that
\[
Q(z,F_L(z))=0.
\]
Thus $F_L(z)$ is algebraic over $\Q(z)$.
\end{proof}

The finite-algebra condition is \emph{sufficient}, not necessary for the
algebraicity of an individual target series. In particular, the general
algebraicity theorem for unambiguous context-free languages does not
require proving that every chosen presentation of its DSV equations
satisfies this stronger condition. The criterion concerns the system,
not the grammatical complexity of the language it describes.

\subsection{Beyond context-free grammars: a rational indexed-language example}

Consider the classical non-context-free language
\[
L_{abc}=\{a^n b^n c^n:n\geq 1\}.
\]
It is context-sensitive and, more specifically, indexed. Its length
generating function is rational, despite the fact that the language is not
context-free.

\paragraph{A context-sensitive grammar.}
One presentation uses the productions
\[
\begin{aligned}
S  &\to aSBC \mid aBC,\\
CB &\to BC,\\
aB &\to ab,\qquad bB\to bb,\\
bC &\to bc,\qquad cC\to cc.
\end{aligned}
\]
For any $n\geq1$, the first productions give
\[
S\Rightarrow^*a^n(BC)^n.
\]
The swaps $CB\to BC$ can collect the markers into
$a^nB^nC^n$, after which the context-sensitive terminal rules produce
$a^nb^nc^n$. Conversely, a successful terminal derivation has the same
numbers of $a$'s, $b$'s, and $c$'s, in that order.

\emph{The standard DSV substitution must not be applied directly to this
grammar.} For example, the context-sensitive rule $aB\to ab$ does not
imply $zB(z)=z^2$: rewriting $B$ is only permitted in the specified
left context. Moreover, the swaps $CB\to BC$ are essential for successful
derivations; they cannot simply be omitted from a counting system.

\paragraph{An indexed grammar and stack-indexed generating functions.}
An indexed grammar for the same language is
\[
\begin{aligned}
S[\sigma]&\to T[g\sigma],\\
T[\sigma]&\to T[f\sigma]\mid A[\sigma]B[\sigma]C[\sigma],\\
A[f\sigma]&\to aA[\sigma],& A[g\sigma]&\to a,\\
B[f\sigma]&\to bB[\sigma],& B[g\sigma]&\to b,\\
C[f\sigma]&\to cC[\sigma],& C[g\sigma]&\to c.
\end{aligned}
\]
Here $g$ marks the bottom of the index stack, and the expansion of $T$
copies the same stack into all three branches. Starting from
$S[\varepsilon]$, the stacks encountered after the initial rule have
the form $f^kg$, for $k\geq0$.

Let $A_k(z)$, $B_k(z)$, $C_k(z)$, and $T_k(z)$ count terminal words
derivable from $A[f^kg]$, $B[f^kg]$, $C[f^kg]$, and $T[f^kg]$,
respectively. The three terminal-producing branches satisfy
\[
\begin{aligned}
A_0(z)&=B_0(z)=C_0(z)=z,\\
A_{k+1}(z)&=zA_k(z),\\
B_{k+1}(z)&=zB_k(z),\\
C_{k+1}(z)&=zC_k(z).
\end{aligned}
\]
Consequently,
\[
A_k(z)=B_k(z)=C_k(z)=z^{k+1}.
\]
The two alternatives for $T[f^kg]$ lead to the
\emph{stack-indexed} recurrence
\[
T_k(z)=T_{k+1}(z)+A_k(z)B_k(z)C_k(z)
      =T_{k+1}(z)+z^{3(k+1)}.
\]
The contributions have different lengths, so no terminal word is
double-counted. Iterating gives
\[
T_0(z)=T_N(z)+\sum_{k=0}^{N-1}z^{3(k+1)}.
\]
Every word counted by $T_N$ has length at least $3(N+1)$; hence
$T_N(z)\to0$ coefficientwise in the formal-power-series topology.
Since $F_{abc}(z)=T_0(z)$, we obtain
\[
\boxed{
F_{abc}(z)=\sum_{n\geq1}z^{3n}=\frac{z^3}{1-z^3}.
}
\]

This is a DSV-\emph{type} derivation using an infinite family of
stack-indexed series, not a finite polynomial DSV system obtained
by treating context-sensitive rewrite rules as context-free.
The indexed-grammar approach is related to the framework developed
by Adams, Freden, and Mishna \cite{AdamsFredenMishna}.

This example anticipates the non-converse results of Section~3. In particular,
\[
\boxed{\text{grammatical complexity}\not\equiv
\text{complexity of the length generating function}.}
\]

\section{Why Generating Functions Do Not Determine Grammar Classes}

A regular language has a rational length generating function, and an
unambiguous context-free language has an algebraic one. The reverse
implications fail for individual languages:
\[
F_L(z)\text{ rational}\not\Longrightarrow L\text{ regular},
\qquad
F_L(z)\text{ algebraic}\not\Longrightarrow
L\text{ unambiguous context-free}.
\]
The example of $L_{abc}$ in Section~2 already defeats the first converse;
the following diagonal construction shows that even the indexed-language
class cannot be inferred from a rational generating function.

\begin{theorem}[Rational length growth does not imply indexed structure]
\label{thm:diag}
There exists a recursive language \(L_{\mathrm{diag}}\subseteq\{0,1\}^*\) such that
\begin{enumerate}[label=(\roman*)]
\item \(L_{\mathrm{diag}}\) is recursive but not indexed, and therefore no indexed grammar generates it;
\item \(|L_{\mathrm{diag}}\cap\{0,1\}^n|=1\) for every \(n\ge0\);
\item \(F_{L_{\mathrm{diag}}}(z)=1/(1-z)\).
\end{enumerate}
\end{theorem}

\begin{proof}
Enumerate all indexed grammars effectively as \(G_0,G_1,G_2,\dots\).
Membership in an indexed language is decidable.
For each \(i\ge0\) let \(B_{i+1}=\{0,1\}^{i+1}\).
If \(B_{i+1}\setminus L(G_i)\) is non-empty, choose its lexicographically first element \(w_{i+1}\); otherwise set \(w_{i+1}=0^{i+1}\).
Put \(w_0=\varepsilon\) and define \(L_{\mathrm{diag}}=\{w_n:n\ge0\}\).
Exactly one word of each length is selected, so the generating function is \(1/(1-z)\).
The language is recursive because membership of a word of length \(n\) reduces to the computation of the single word \(w_n\) by finitely many decidable tests.
Suppose for contradiction that \(L_{\mathrm{diag}}=L(G_i)\) for some \(i\).
If \(B_{i+1}\setminus L(G_i)\) is non-empty then \(w_{i+1}\notin L(G_i)\), a contradiction.
If \(B_{i+1}\subseteq L(G_i)\) then \(G_i\) contains \(2^{i+1}\) words of length \(i+1\), while \(L_{\mathrm{diag}}\) contains only one, again a contradiction.
Hence \(L_{\mathrm{diag}}\) is not indexed, so no indexed grammar generates it.
\end{proof}

The same rational series \(1/(1-z)\) is realised by regular, context-free, indexed and highly non-indexed languages.
Consequently,
\begin{equation}
\boxed{F_L\text{ algebraic}\quad\not\Longrightarrow\quad
L\text{ is indexed}.}
\end{equation}
The conclusion concerns indexed \emph{grammatical descriptions}, not algebraic
relations for the length series: here the same series still satisfies
$(1-z)F_L(z)-1=0$.

Recursive languages with transcendental length generating functions may
nevertheless admit effective finite grammatical descriptions beyond finite
algebraic DSV systems. The indexed-grammar extension of DSV developed by
Adams, Freden and Mishna can yield transcendental, lacunary series with
natural boundaries through infinite stack-indexed generating-function
systems derived from finite indexed grammars \cite{AdamsFredenMishna}.

This construction separates two questions: whether a language has a
particular type of grammar, and whether its length generating function admits
an algebraic description. The same coefficient sequence may correspond to
languages of very different grammatical complexity. Consequently, the
analytic class of $F_L$ cannot identify the Chomsky class of $L$.

\section{Computability of Length-Counting Sequences}

\subsection{Length counting and computability}

The ordinary generating function retains only the sequence
$(\ell_n)_{n\geq0}$. This raises a natural computability question:
Does computability of $(\ell_n)$ imply decidability of $L$? For recursively
enumerable languages the answer is yes.

\begin{theorem}[Computability obstruction]
\label{thm:III}
Let $L\in\RE$. If the sequence
$\ell_n=\#\{w\in L:|w|=n\}$ is computable, then $L\in\Rec$.
Consequently,
\[
L\in\RE\setminus\Rec
\quad\Longrightarrow\quad
(\ell_n)_{n\geq0}\text{ is not computable}.
\]
\end{theorem}

\begin{proof}
Assume that $L$ is recursively enumerable and that $\ell_n$ is computable.
Given a word $w$ of length $n$, first compute $\ell_n$. If $\ell_n=0$,
then $w\notin L$, and we reject. Suppose instead that $\ell_n>0$.
Because $L$ is recursively enumerable, there is an enumeration of its
words. Run this enumeration until exactly $\ell_n$ distinct words of
length $n$ have appeared.

At that point every word of $L$ of length $n$ has appeared, because by
definition there are exactly $\ell_n$ such words. We can therefore decide
whether the input word $w$ is among them. Hence $L$ is recursive.
\end{proof}

This theorem gives the precise form of the computability obstruction. For
example, the language
\[
K=\{\langle M\rangle:M\text{ halts on the empty input}\}
\]
is recursively enumerable but not recursive. Therefore its length-counting
sequence cannot be computable.

Thus there cannot exist an effective procedure which, from a finite
description of such a language in an effective r.e. grammar formalism,
produces all coefficients of its ordinary length generating function.
The point is subtle but important:
\begin{equation}
\boxed{L\in\RE\setminus\Rec\quad\Longrightarrow\quad
F_L(z)=\sum\ell_nz^n\text{ has a non-computable coefficient sequence}.}
\end{equation}
This does not mean that the formal power series cannot be written
symbolically. It means that there is no algorithm that computes $\ell_n$
for arbitrary $n$.

\subsection{Recursive languages}

The converse direction is immediate.

\begin{theorem}[Computability of coefficients]
\label{thm:IV}
If $L\in\Rec$, then the coefficient sequence $(\ell_n)_{n\geq0}$ of
$F_L(z)$ is computable.
\end{theorem}

\begin{proof}
Given $n$, enumerate the finite set $\Sigma^n$. Since $L$ is recursive,
membership of each word can be decided. Counting the accepted words gives
$\ell_n$.
\end{proof}

Therefore every recursive language has a computable ordinary length
generating series in the coefficientwise sense, regardless of its analytic
classification. This gives the fundamental computability dichotomy
\[
\begin{array}{c}
L\in\Rec\\[2mm]\Downarrow\\[2mm]
(\ell_n)\text{ computable}
\end{array}
\qquad
\begin{array}{c}
L\in\RE\setminus\Rec\\[2mm]\Downarrow\\[2mm]
(\ell_n)\text{ non-computable}.
\end{array}
\]
The distinction is independent of algebraicity.

\subsection{$D$-finiteness as a decidability criterion}

The preceding results give a direct connection between a differential
property of the length generating function and language computability.

\begin{corollary}[$D$-finite length series force decidability]
\label{cor:dfinite-decidable}
Let $L\subseteq\Sigma^*$ be a recursively enumerable language over a finite
alphabet. If its ordinary length generating function
$F_L(z)=\sum_{n\geq0}\ell_nz^n$ is $D$-finite over $\Q(z)$, then $L$ is
recursive. Equivalently,
\[
\boxed{
 L\in\RE\setminus\Rec
 \quad\Longrightarrow\quad
 F_L(z)\text{ is not $D$-finite over }\Q(z).
}
\]
In particular, the length generating function of a non-recursive
recursively enumerable language cannot be algebraic or rational over
$\Q(z)$.
\end{corollary}

\begin{proof}
The coefficients of a $D$-finite power series over $\Q(z)$ satisfy a
linear recurrence with polynomial coefficients (they form a
$P$-recursive sequence) \cite{Stanley}. Thus, for some $r\geq1$ and
polynomials $q_0,\ldots,q_r\in\Q[n]$, with $q_r\ne0$,
\[
\sum_{j=0}^{r}q_j(n)\ell_{n+j}=0
\]
for all sufficiently large $n$. Since a nonzero polynomial has only
finitely many roots, there is an $N$ such that $q_r(n)\ne0$ for all
$n\geq N$. For these $n$,
\[
\ell_{n+r}=-\frac{1}{q_r(n)}
               \sum_{j=0}^{r-1}q_j(n)\ell_{n+j}.
\]
The recurrence and the finitely many initial coefficients determine an
algorithm computing every $\ell_n$. By Theorem~\ref{thm:III},
$L\in\Rec$. The contrapositive gives the stated obstruction;
algebraic and rational series are $D$-finite.
\end{proof}

The converse is not claimed: decidability ensures only computable
coefficients. The argument is non-uniform; it does not extract a
recurrence and initial values from an arbitrary r.e. grammar.

\section{Analytic Properties and Singularities}

\subsection{Coefficient growth and positive radius}

For every language $L\subseteq\Sigma^*$,
$0\leq\ell_n\leq|\Sigma|^n$. Hence
\[
\limsup_{n\to\infty}\ell_n^{1/n}\leq|\Sigma|,
\]
and therefore $R\geq1/|\Sigma|>0$. Thus the ordinary length generating
function of every language over a finite alphabet has a positive radius of
convergence.

If $L$ is infinite, then $\ell_n>0$ for infinitely many $n$, and hence
$\limsup_{n\to\infty}\ell_n^{1/n}\geq1$. Therefore
\begin{equation}
0<R\leq1
\qquad\text{for every infinite language }L.
\label{eq:radius-language}
\end{equation}

The Cauchy--Hadamard formula gives
\[
R^{-1}=\limsup_{n\to\infty}\ell_n^{1/n}.
\]
Thus the exponential scale of the length counts determines the radius;
no membership or parsing information can be recovered from this fact alone.
If, additionally, $R=1$ and $\ell_n\geq1$ for every $n\geq0$, then
$F_L(r)\geq(1-r)^{-1}$ for $0<r<1$, so $F_L(r)\to\infty$ as $r\to1^-$. This is
a special case of Pringsheim's theorem.

\subsection{Infinite languages and productive recursion}

We now formulate the recursion statement more carefully.

\begin{definition}[Productive recursive cycle in a CFG]
Let $G=(N,\Sigma,P,S)$ be a finite context-free grammar. A
\emph{reachable, terminating, productive recursive cycle} consists of a
nonterminal $A\in N$ and terminal strings $p,q,u,v,w\in\Sigma^*$ such that
\[
S\Rightarrow^*pAq,\qquad
A\Rightarrow^+uAv,\qquad
A\Rightarrow^*w,\qquad |uv|>0.
\]
The derivation $A\Rightarrow^+uAv$ corresponds to a closed directed walk
through $A$ in the nonterminal-dependency graph and can be repeated.
The condition $|uv|>0$ ensures that the recursive traversal, rather than
some unrelated branch, increases the terminal length.
\end{definition}

\begin{theorem}[Singularity and productive recursion for finite CFGs]
\label{thm:V}
Let $L\subseteq\Sigma^*$ be any infinite language over a finite alphabet,
and let $F_L(z)=\sum_{n\geq0}\ell_nz^n$. Then:
\begin{enumerate}[label=(\roman*)]
\item $F_L$ has a finite positive radius of convergence $R$;
\item $F_L$ has a singularity at $z=R$.
\end{enumerate}
If, in addition, $L=L(G)$ for a finite context-free grammar $G$, then
\begin{enumerate}[label=(\roman*),start=3]
\item $G$ contains a reachable, terminating, productive recursive cycle
in the sense of the preceding definition.
\end{enumerate}
\end{theorem}

\begin{proof}
Since $\Sigma$ is finite, $\ell_n\leq|\Sigma|^n$, so $R>0$.
Infinitude of $L$ implies that $\ell_n>0$ for infinitely many $n$, hence
$R\leq1$. Thus $0<R\leq1$. Because $\ell_n\geq0$, Pringsheim's theorem
implies that $z=R$ is singular
\cite{Pringsheim,FlajoletSedgewick}.

For (iii), suppose $L=L(G)$ for a finite CFG. For each terminal word,
choose a derivation tree with the smallest possible number of
nonterminal nodes. If one nonterminal $A$ appeared twice on the same
root-to-leaf path, the intervening part of the tree would give a
recursion $A\Rightarrow^+uAv$ after fully expanding all side branches
into terminal words. If $|uv|>0$, this would be a reachable,
terminating, productive recursive cycle. If $|uv|=0$, the intervening
part could instead be deleted, producing a smaller derivation tree
for the same terminal word, contrary to minimality.

Thus, in the absence of any such productive recursive cycle, each
minimal derivation tree has no repeated nonterminal along a path.
Its depth is therefore bounded by $|N|$. Since the grammar has finitely
many productions, their right-hand sides have uniformly bounded
length, so the terminal yields of these trees have uniformly bounded
length. This contradicts the fact that $L$ is infinite. Hence a
productive recursive cycle must exist.
\end{proof}

The theorem establishes a necessary relationship between the existence
of productive grammatical recursion and the existence of a finite
positive-real singularity. For every infinite language generated by a
finite context-free grammar, a reachable and terminating growing
self-embedding must occur. Independently, its nonnegative integer
length-counting coefficients imply that the positive radius of
convergence is a singular point.

This does \emph{not} establish a one-to-one correspondence between
individual recursive productions and singularities. Such a stronger
statement would require additional hypotheses and a precise definition
of the proposed correspondence. Nor does productive recursion by itself
imply a natural boundary. More precisely, the two implications are
\[
\text{infinite CFL generated by a finite CFG}
\Longrightarrow\text{productive recursion},
\]
while independently
\[
\text{infinite language over a finite alphabet}
\Longrightarrow 0<R\leq1
\Longrightarrow \text{a positive-real singularity}.
\]
For example, $S\to aS\mid\varepsilon$ generates $\{a\}^*$, and
$F(z)=1/(1-z)$ has a simple pole at $z=1$. By contrast, a finite language
has a polynomial generating function and therefore no finite
singularity.

\subsection{Integer coefficients: the Polya--Carlson dichotomy}

The arithmetic condition that the coefficients be integers imposes a much
stronger restriction when the radius of convergence equals $1$.

\begin{theorem}[Polya--Carlson]
\label{thm:PC}
Let $F(z)=\sum_{n\geq0}a_nz^n$, $a_n\in\Z$, and suppose that $F$ has radius
of convergence $1$. Then either
\begin{enumerate}[label=(\roman*)]
\item $F(z)$ is rational, or
\item the unit circle $|z|=1$ is a natural boundary.
\end{enumerate}
\end{theorem}

The theorem was conjectured by Polya and proved by Carlson
\cite{Polya,Carlson}. It is one of the basic connections between arithmetic
properties of coefficients and analytic continuation.

\begin{corollary}
If $F(z)\in\Z[[z]]$ has radius $1$ and is $D$-finite, then $F(z)$ is
rational.
\end{corollary}

\begin{proof}
Suppose, for contradiction, that $F$ is not rational. Since
$F\in\Z[[z]]$ has radius of convergence $1$, the Polya--Carlson theorem
implies that the unit circle $|z|=1$ is a natural boundary of $F$.

On the other hand, $D$-finiteness means that $F$ satisfies a nontrivial
linear differential equation
\[
  p_r(z)F^{(r)}(z)+\cdots+p_1(z)F'(z)+p_0(z)F(z)=0,
\]
where $p_0,\ldots,p_r\in\C[z]$ and $p_r\neq 0$. At any point $z_0$
with $p_r(z_0)\neq 0$, the equation is a linear differential equation
with analytic coefficients and nonvanishing leading coefficient. The
analytic-continuation theorem for linear differential equations therefore
allows $F$ to be continued analytically through such a point of the unit
circle. At least one such point exists because the nonzero polynomial
$p_r$ has only finitely many zeros.

This contradicts the assertion that the entire unit circle is a natural
boundary. Hence $F$ is rational.
\end{proof}

Thus, at radius $1$,
\[
\boxed{\text{integer coefficients}+\text{$D$-finite}
\Longrightarrow\text{rational}.}
\]
This is an especially useful restriction when generating functions arise
from formal languages.

\subsection{Algebraic and \texorpdfstring{$D$}{D}-finite examples}

\paragraph{Algebraic functions.}

A standard example is the Catalan generating function
\[
C(z)=\sum_{n\geq0}C_nz^n
=\frac{1-\sqrt{1-4z}}{2z},
\qquad
C_n=\frac1{n+1}\binom{2n}{n}.
\]
Its radius of convergence is $R=1/4$, and $z=1/4$ is an algebraic branch
point.

The Dyck language is generated by
$S\to\varepsilon\mid(S)S$. The number of Dyck words of length $2n$ is
$C_n$. Consequently its ordinary word-length generating function is
\[
F_D(z)=\sum_{n\geq0}C_nz^{2n}=C(z^2)
=\frac{1-\sqrt{1-4z^2}}{2z^2}.
\]
Thus the precise statement is
\[
\boxed{\text{Dyck language}\longrightarrow
\text{algebraic semilength generating function}.}
\]
The distinction between word length and semilength is important.

\paragraph{$D$-finite transcendental functions.}

There are also positive-integer-coefficient $D$-finite functions that are
not algebraic. Consider
\[
F(z)=\sum_{n\geq0}\binom{2n}{n}^2z^n.
\]
The coefficients satisfy
\[
(n+1)^2a_{n+1}=4(2n+1)^2a_n,
\]
so $F$ satisfies
\begin{equation}
z(1-16z)F''(z)+(1-32z)F'(z)-4F(z)=0.
\label{eq:binomial-square-ode}
\end{equation}
It is therefore $D$-finite. On the other hand, it is a hypergeometric
function,
\[
F(z)={}_2F_1\!\left(\frac12,\frac12;1;16z\right),
\]
and is expressible in terms of the complete elliptic integral of the first
kind. The corresponding analytic function is transcendental
\cite{Flajolet,Stanley}.

A formal-language interpretation is also available. Flajolet considers
languages encoding two-dimensional walks and obtains an explicitly defined
language $I$ whose words of length $2n$ are counted by
\[
[z^{2n}]I(z)=\binom{2n}{n}^2,
\qquad [z^{2n+1}]I(z)=0
\]
\cite{Flajolet}. Thus the ordinary \emph{word-length} generating
function of this language is
\[
I(z)=\sum_{n\geq0}\binom{2n}{n}^2 z^{2n}=F(z^2).
\]
In contrast, $F(z)$ counts \emph{semilength} $n$ (half the number of
steps), not ordinary word length. In particular, their radii of
convergence differ: $R_F=1/16$ and $R_I=1/4$. Both series are
$D$-finite and transcendental. This provides a concrete bridge
\[
\text{formal language}\longrightarrow\text{integer coefficient sequence}
\longrightarrow D\text{-finite transcendental function}.
\]
The important point is that
\[
\boxed{D\text{-finite}\not\Longrightarrow\text{algebraic}.}
\]

\subsection{Siegel's \texorpdfstring{$G$}{G}-functions}

Not every transcendental function has a natural boundary. An important
intermediate class consists of $G$-functions, introduced by Siegel in his
work on Diophantine approximation \cite{Siegel}.

\begin{definition}[Siegel $G$-function]
Let $F(z)=\sum_{n\geq0}a_nz^n$, $a_n\in K$, where $K$ is a number field.
The series is a $G$-function if:
\begin{enumerate}[label=(\roman*)]
\item $F(z)$ satisfies a non-zero linear differential equation with
coefficients in $K(z)$;
\item the absolute values of all Galois conjugates of $a_n$ are bounded by
$C^{n+1}$ for some constant $C>0$;
\item the least common denominator of $a_0,a_1,\ldots,a_n$ is at most
$C^{n+1}$;
\item the series has positive radius of convergence.
\end{enumerate}
\end{definition}

For integer coefficients, the denominator condition is automatic, while the
growth condition remains nontrivial. By Eisenstein's theorem on algebraic
power series, every algebraic power series with rational coefficients is a
$G$-function; more generally, this holds for algebraic power series with
algebraic coefficients. There are also transcendental $G$-functions,
including the hypergeometric series $F(z)$ considered above: its integer
coefficients grow at most exponentially, and it is $D$-finite
\cite{Siegel,Stanley}.

Thus the analytic hierarchy should not be reduced to the dichotomy
\[
\text{algebraic}\quad\text{versus}\quad\text{natural boundary}.
\]
In the following strict inclusions, all series lie in
$\overline{\mathbb Q}[[z]]$, and algebraicity and $D$-finiteness are
understood over $\overline{\mathbb Q}(z)$:
\[
\boxed{
\begin{gathered}
\text{algebraic power series over }\overline{\mathbb Q}(z)
\;\subsetneq\; G\text{-functions}\;\subsetneq\; \\[-1pt]
D\text{-finite power series over }\overline{\mathbb Q}(z).
\end{gathered}}
\]
The first inclusion is strict by the transcendental $G$-function $F$ above.
The second is strict because $e^z$ is $D$-finite over $\mathbb Q(z)$,
but is not a $G$-function: the denominators $n!$ of its Taylor
coefficients are not bounded exponentially.

\subsection{Natural boundaries and lacunary examples}

\paragraph{Infinite singularity sets.}

A different phenomenon occurs when the generating function has infinitely
many singularities. A classical example is the Lambert series
\[
L(z)=\sum_{n\geq1}\frac{z^n}{1-z^n}
=\sum_{m\geq1}d(m)z^m,
\]
where $d(m)$ denotes the number of positive divisors of $m$.
This identity follows from
\[
\sum_{n\geq1}\frac{z^n}{1-z^n}
=\sum_{n\geq1}\sum_{k\geq1}z^{nk}.
\]
The divisor-function coefficients arise naturally in formal-language
enumeration. Flajolet gives an explicit language construction whose
generating function is a divisor-type series and uses its singular behaviour
to prove transcendence and inherent ambiguity \cite{Flajolet}.

The corresponding language should not be confused with the simple regular
language $\{a^kb^m:k,m\geq1\}$, whose ordinary length generating function
is merely $z^2/(1-z)^2$. The divisor function arises from multiplicative
length constraints, not from ordinary concatenation of two independent
blocks.

\paragraph{Hadamard-gap series.}

Let $n_0<n_1<n_2<\cdots$ satisfy a Hadamard gap condition
$n_{k+1}/n_k\geq q>1$. Then the lacunary series
\[
F(z)=\sum_{k\geq0}z^{n_k}
\]
has the unit circle as a natural boundary \cite{HadamardGap,FlajoletSedgewick}.
The simplest example is
\[
F(z)=\sum_{k\geq0}z^{2^k}=z+z^2+z^4+z^8+\cdots.
\]
Since $2^{k+1}/2^k=2$, the Hadamard gap theorem applies. Consequently
$|z|=1$ is a natural boundary. Hence $F$ is transcendental. In particular
it cannot be algebraic, because an algebraic function has only finitely many
branch points on its compact Riemann surface.

The corresponding unary language is
\[
L_{\mathrm{pow}}=\{a^{2^k}:k\geq0\}.
\]
Its ordinary length generating function is exactly
$F_{L_{\mathrm{pow}}}(z)=\sum_{k\geq0}z^{2^k}$. The language is not regular
and in fact is not context-free; it lies in the indexed-language class
\cite{Aho,HopcroftUllman}. Thus an arithmetic constraint on allowable
lengths can produce a natural boundary.

\paragraph{Factorial gaps.}

Another example is
\[
F(z)=\sum_{n\geq1}(2z)^{n!}
=\sum_{n\geq1}2^{n!}z^{n!}.
\]
The sum starts at $n=1$, since $0!=1!=1$ and length $1$ must be
counted only once. The radius of convergence is $R=1/2$. Since
$(n+1)!/n!=n+1\to\infty$, the exponents satisfy an even stronger gap
condition than required by Hadamard's theorem. Therefore $|z|=1/2$ is a
natural boundary.

The associated unary language
$L_{\mathrm{fact}}=\{a^{n!}:n\geq0\}$ has
$\sum_{n\geq1}z^{n!}$ as its length generating function. If one wishes to
obtain radius $1/2$ while retaining integer coefficients, one may instead
use the binary language
\[
L'_{\mathrm{fact}}=\{w\in\{0,1\}^*:|w|=n!\text{ for some }n\},
\]
whose generating function is
$\sum_{n\geq1}2^{n!}z^{n!}$. This example shows that transcendence and
natural boundary can coexist with a completely recursive length constraint.

For comparison, $(1-z)^{-1}$ has an isolated pole,
the Catalan series has an algebraic branch point, and
$\sum_{n\geq0}z^{2^n}$ has a natural boundary. In particular,
\[
\text{natural boundary}\Longrightarrow\text{transcendental},
\]
but not conversely: the $D$-finite transcendental walk series described
above has only finitely many possible finite singularities.

\subsection{A unified analytic and arithmetic classification}

The four-way division into rational, algebraic non-rational, $D$-finite
transcendental, and non-$D$-finite transcendental series is compatible with
the nested hierarchy of Siegel $G$-functions.  Indeed, the $G$-function
condition refines the $D$-finite part: every algebraic power series over
$\Q(z)$ is a $G$-function, and transcendental $D$-finite series can be
$G$-functions or non-$G$-functions.  Figure~\ref{fig:classification}
combines the two viewpoints in a single tree.

\begin{figure}[htbp]
\centering
\begin{forest}
for tree={draw,rounded corners,align=center,font=\footnotesize,
  parent anchor=south,child anchor=north,
  edge={semithick},inner sep=3pt,l sep=10pt,s sep=6pt}
[Convergent power series over $\Q$, font=\normalsize, inner sep=5pt
  [$D$-finite
    [Siegel $G$-functions
      [Algebraic
        [Rational
          [{\scriptsize Poles (or none)},draw=none,edge={draw=none}]
        ]
        [Algebraic\\non-rational
          [{\scriptsize Branch points possible},draw=none,edge={draw=none}]
        ]
      ]
      [Transcendental\\$G$-functions
        [{\scriptsize Elliptic / hypergeometric},draw=none,edge={draw=none}]
      ]
    ]
    [$D$-finite\\non-$G$
      [{\scriptsize Example: $e^z$ (entire)},draw=none,edge={draw=none}]
    ]
  ]
  [Non-$D$-finite\\(transcendental)
    [{\scriptsize Natural boundary possible,\\not necessary},draw=none,edge={draw=none}]
  ]
]
\end{forest}
\caption{A unified classification of convergent power series over $\Q$.
The boxed nodes denote nested function classes; the unboxed leaves give
\emph{illustrative}, not defining, analytic behavior.  The four broad
(disjoint) types are rational, algebraic non-rational, $D$-finite
transcendental (the two transcendental $D$-finite branches), and
non-$D$-finite transcendental.  Siegel $G$-functions consist of the
algebraic branch together with the transcendental $G$-function branch.}
\label{fig:classification}
\end{figure}
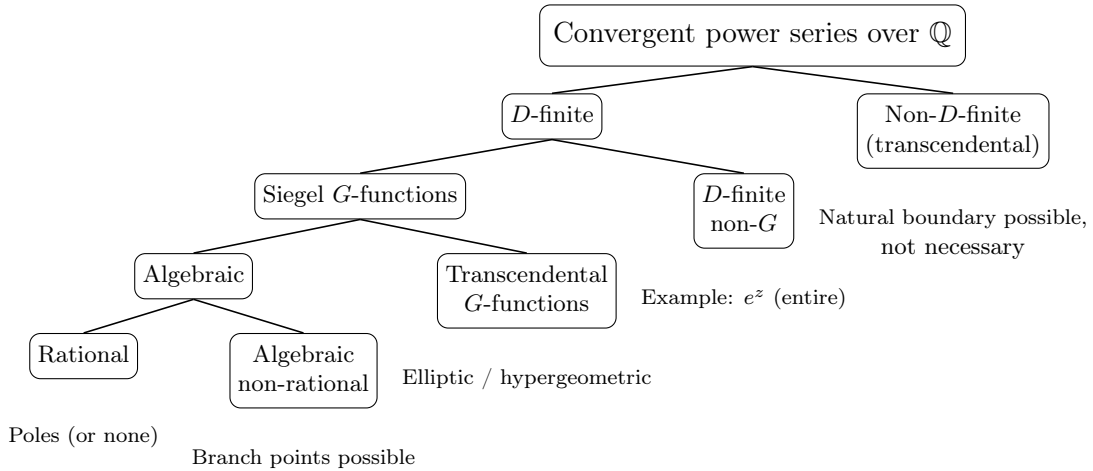

The distinction between $G$-functions and general $D$-finite functions
matters over $\Q[[z]]$: for example, $e^z$ is $D$-finite but not a
$G$-function because its coefficients have factorial denominators.  In
contrast, an ordinary length generating function $F_L(z)=\sum \ell_nz^n$
of a language over a finite alphabet $\Sigma$ satisfies
$0\leq\ell_n\leq|\Sigma|^n$ with $\ell_n\in\Z$.  Hence 
\[
\boxed{F_L\text{ is $D$-finite over $\Q(z)$}
\quad\Longleftrightarrow\quad F_L\text{ is a Siegel $G$-function}.}
\]
The forward implication follows from the integer coefficients, their
exponential bound, and $D$-finiteness; the reverse implication holds by
definition.  Thus the $D$-finite non-$G$ branch of the unified tree
\emph{cannot occur} for ordinary length generating functions over finite
alphabets.  Singularity behavior is a separate refinement: a
non-$D$-finite function need not have a natural boundary.

\section{Discussion and Conclusion}

The paper exhibits three complementary principles. First, strong grammatical
hypotheses provide effective counting descriptions: regular languages have
rational length generating functions, unambiguous context-free languages have
algebraic ones, and stack-indexed recurrences describe particular indexed
languages. The finite-algebra criterion in Section~2 is a sufficient
algebraic condition on a chosen polynomial system, not a characterization of
all grammar-derived systems.

Second, the reverse inference loses information. The diagonal construction
in Section~3 exhibits a decidable non-indexed language with the rational
length generating function $1/(1-z)$. Consequently, rationality or
algebraicity alone cannot determine the grammatical class or recover a unique
grammar.

Third, computability imposes restrictions beyond analytic simplicity.
Within the recursively enumerable languages, the
length-counting sequence is computable exactly when membership is decidable.
Consequently, $D$-finiteness of the length generating function forces
decidability within the r.e. class; non-recursive r.e. languages cannot
have $D$-finite length generating functions over $\Q(z)$.
Thus it is important to distinguish a formally associated series from a
coefficient sequence that is algorithmically accessible. In particular, an
arbitrary finite unrestricted grammar cannot, in general, be translated
effectively into an algorithm computing every length coefficient.

The analytic results give further one-way constraints. Infinite languages
over finite alphabets have a positive finite radius and a positive-real
singularity; finite CFGs generating infinite languages require productive
self-embedding. For integer coefficients at radius $1$, the Polya--Carlson
theorem yields a rational/natural-boundary dichotomy, while $D$-finite
functions cannot have a natural boundary. Algebraic, $G$-function, and
$D$-finite distinctions describe genuine analytic structure, but do not
invert the many-to-one map $L\mapsto F_L$.

The main conclusion is therefore an asymmetry: specified grammars and
computational constraints can force properties of length generating
functions, whereas those properties alone rarely characterize the language.

\appendix
\section{Examples of Integer-Coefficient Generating Functions}


Let
\[
F(z)=\sum_{n=0}^{\infty}a_nz^n,
\qquad a_n\in\Z_{\geq0},
\]
with radius of convergence $R>0$.
The following table summarizes representative examples.

\begin{center}
\footnotesize
\renewcommand{\arraystretch}{1.45}
\begin{tabularx}{\linewidth}{@{}>{\raggedright\arraybackslash}p{2.0cm}
>{\raggedright\arraybackslash}p{3.9cm}
>{\raggedright\arraybackslash}p{3.1cm}
>{\raggedright\arraybackslash}X@{}}
\toprule
\textbf{Class} & \textbf{Generating function} &
\textbf{Analytic structure} & \textbf{Language / combinatorial source}\\
\midrule
Rational & $\displaystyle\frac1{1-z}$ & Simple pole at $1$ & $\{a^n:n\geq0\}$\\
Rational & $\displaystyle\frac1{1-2z}$ & Simple pole at $1/2$ & All binary words\\
Algebraic & $\displaystyle\frac{1-\sqrt{1-4z}}{2z}$ & Square-root branch point at $1/4$ & Dyck paths (semilength)\\
$D$-finite transcendental & $\displaystyle\sum_{n\geq0}\binom{2n}{n}^2z^n$ & Finite singular locus; elliptic / hypergeometric & Two-dimensional walks (semilength; the ordinary GF is $F(z^2)$)\\
Lambert / divisor & $\displaystyle\sum_{n\geq1}\frac{z^n}{1-z^n}$ & Infinitely many boundary singularities & Divisor-type combinatorial languages\\
Natural boundary & $\displaystyle\sum_{n\geq0}z^{2^n}$ & Unit circle is a natural boundary & $\{a^{2^n}:n\geq0\}$\\
Natural boundary & $\displaystyle\sum_{n\geq1}z^{n!}$ & Unit circle is a natural boundary & $\{a^{n!}:n\geq0\}$\\
\bottomrule
\end{tabularx}
\end{center}

The table is deliberately not presented as an exhaustive classification. In
particular, the classes of transcendental functions with infinitely many
isolated singularities, natural boundaries, $D$-finite functions,
$G$-functions, and more general recursively generated functions overlap in
several ways.

\section{A More Precise Language--Series Dictionary}

The examples above support the following hierarchy of implications:
\[
\begin{array}{ccccc}
\text{regular} &\Longrightarrow& \text{rational} &\Longrightarrow& \text{algebraic}\\[3pt]
&&&&\\[-7pt]
&&\text{unambiguous CFL}&\Longrightarrow&\text{algebraic}\\[3pt]
&&&&\\[-7pt]
&&\text{indexed}&\Longrightarrow&\text{more general functional equations}
\end{array}
\]
but there is no converse chain
\[
\text{analytic class}\Longrightarrow\text{unique language class}.
\]

For recursively enumerable languages, there is also a one-way implication:
\[
\boxed{L\in\RE,\quad F_L\text{ $D$-finite over }\Q(z)
\quad\Longrightarrow\quad L\in\Rec.}
\]

In particular, $1/(1-z)$ can encode a regular language, an indexed language,
or a recursive language outside the indexed hierarchy. Similarly,
$\sum_{n\geq0}z^{2^n}$ is transcendental with a natural boundary, but the
language $\{a^{2^n}:n\geq0\}$ is recursive.

Therefore the correct conceptual interpretation is not
\[
\text{more complicated grammar}\Longrightarrow\text{more complicated analytic function}
\]
as an equivalence, but rather
\[
\boxed{\text{specific grammatical mechanisms impose analytic constraints},}
\]
while
\[
\boxed{\text{analytic behaviour alone generally does not determine grammatical complexity}.}
\]
This asymmetry is the principal structural conclusion of the paper.


\end{document}